\documentclass[12pt]{article}
\usepackage{amsmath,amssymb,amsthm,geometry,bbm,calrsfs,graphicx}
\usepackage{authblk}
\usepackage{dcolumn}
\usepackage[utf8]{inputenc}
\usepackage[T1]{fontenc}
\usepackage{rsfso}
\usepackage{color}
\usepackage{hyperref}
\usepackage{fourier}
\newcommand{\ket}[1]{\left|#1\right\rangle}
\newcommand{\bra}[1]{\left\langle#1\right|}
\newcommand{\state}[1]{\ket{#1}\bra{#1}}
\newcommand{\expc}[1]{\left\langle #1 \right\rangle}
\newcommand{\hilb}{\mathcal{H}}

\newcommand{\Tr}{\mathrm{Tr}}

\newcommand{\Expc}{\mathbb{E}}

\newcommand{\dint}[2]{{\displaystyle \int_{#1}^{#2}}}

\newcommand{\lind}{\mathcal{L}}

\newcommand{\wt}[1]{\widetilde{#1}}
\newcommand{\ob}[1]{\overline{#1}}

\newcommand{\quadvar}[1]{\left\llbracket#1\right\rrbracket}

\newtheorem{proposition}{Proposition}
\newtheorem{theorem}{Theorem}

\newtheorem{corollary}{Corollary}
\title{Nonlinear functionals of master equation unravelings}
\begin{document}
\author[1]{Dustin Keys}
\affil[1]{Department of Mechanical Engineering, University of North Texas}
\author[2]{Jan Wehr}
\affil[2]{Department of Mathematics and Program in Applied Mathematics, University of Arizona}
\date{\today}
\maketitle
\begin{abstract}
    Unravelings provide a probabilistic representation of solutions of master equations and a method of computation of the density operator dynamics. The trajectories generated by unravelings may also be treated as real---as in the stochastic collapse models. 
 While averages of linear functionals of the unraveling trajectories can be calculated from the master equation, the situation is different for nonlinear functionals, thanks to the corrections with nonzero expected values, coming from the It\^o formula.  Two types of nonlinear functionals are considered here: variance, and entropy. The corrections are calculated explicitly for two types of unravelings, based on Poisson and Wiener processes.  In the case of entropy, these corrections are shown to be negative, expressing the localization introduced by the Lindblad operators.
  
\end{abstract}
\section{Introduction}
In the study of open quantum systems, unravelings of the master equation dynamics have proven themselves to be a useful mathematical tool. Rather than directly integrating the master equation, which governs the dynamics of the reduced system and its respective density operator, one generates realizations of a  stochastic process for a wavevector---the unraveling--- governed by a suitable stochastic differential equation\cite{MCD93,GP92}, such that the ensemble average of the one-dimensional projector generated by the wave vector solves the master equation.
\par In this framework, the unraveling is merely a useful mathematical tool for integrating the master equation---only the ensemble average of the projector plays a role. However, such stochastic equations are obtained naturally from the description of an interaction of the system with the environment.  One such scenario is:  the environment is initially in a reference state; interaction leads to entanglement of the system and its environment; when a measurement is performed on the latter, this affects the system's state.  The environment is reset to the reference state and the whole process is repeated.  In the limit, in which the duration of one step of the process goes to zero, one can show that the system's state satisfies one of the stochastic equations unraveling the master equation (which one, depends on the environmental observable that is measured). The source of randomness is the result of the measurement, distributed according to the Born rule.  Stochastic equations are also postulated in attempts to resolve the quantum measurement problem by introducing a noise-driven collapse process to account for the objectification of outcomes of measurements. In such stochastic collapse models\cite{GRW86,GPR90}, the ontological status of trajectories becomes important, as collapse happens only for single trajectories and not for the evolution of the density operator which obeys a linear master equation \cite{BG00}. 

A natural question to ask is: can we detect the presence of these trajectories, or are all verifiable predictions accounted for in the master equation?
It is thus useful to describe predictions of the quantum trajectories models which go beyond the master equation framework. This work explores the unraveling-specific features of different unravelings, corresponding to the same master equation.  In a recent paper \cite{Petal23}, variance of measurement results associated with quantum trajectories is discussed in experimental context.
\section{GKSL equation and Unravelings}
A full description of the evolution of an open quantum system is given by the joint Schr\"odinger dynamics, $U_t$, on the total Hilbert space $\hilb=\hilb_S\otimes\hilb_E$---the product of the 
system's space, $\hilb_S$, and of the environment's space, $\hilb_E$. The joint state, which is initially a product $\rho(0)=\rho_S\otimes\rho_E$, evolves according to
$$\rho(t)=U_t\rho_S\otimes\rho_EU_t^\dagger,$$
Typically, one is interested primarily in the relative state of the system, which we call $\ob{\rho}(t)$, obtained by taking the trace over the degrees of freedom associated with the environment:
$$\ob{\rho}(t)=\Tr_E\left[U_t\rho_S\otimes\rho_EU_t^\dagger\right].$$
Defined this way, the system's state does not have any intrinsically defined dynamics.  That is, $\ob{\rho}(t)$ at positive times is not determined by $\ob{\rho}(0)$.
To derive an approximate dynamics governing the evolution of $\ob{\rho}(t)$, the {\it Born-Markov approximation} is used. (see e.g. \cite{Sch07,BP02,HR06}). This approximation relies on two assumptions: 
that the system-environment interaction is sufficiently small for the joint state to remain close to a product state (Born approximation), and that the correlation time of the environment is small 
compared to characteristic time scales of the system's evolution. Under the Born-Markov approximation, the reduced density operator $\ob{\rho}(t)$ evolves according to a semigroup $T_t$,
$$\ob{\rho}(t)=T_t[\rho_S],$$
Its infinitesimal generator is the superoperator $\lind[\rho]$
$$\dfrac{d}{dt}T_t[\rho]\big|_{t=0}=\lind[\rho],$$
called the Lindbladian. The term ``superoperator'' refers to the fact that the object evolving under the semigroup $T_t$ is an operator itself---the density operator of the system.
It has been shown by Gorini, Kossakowski and Sudarshan \cite{GKS76}, and by Lindblad \cite{Lind76}, that the general form of this generator is
\begin{equation} 
    \lind[\rho]=-i[H',\rho]-\frac{1}{2}\sum_i\left(L_i^\dagger L_i\rho+\rho L_i^\dagger L_i-2L_i\rho L_i^\dagger\right)
\end{equation}
Here  $H'$ is a self-adjoint operator which is a modified system Hamiltonian.  This modification (called ``relaxation''), results in a change in the energy levels due to the interaction 
with the environment---described as the ``Lamb shift'', by extension of the case of an atom interacting with the electromagnetic field. The $L_i$'s are called Lindblad operators.
In the case when the system's Hilbert space is infinite-dimensional, there may be infinitely many of them and then the sum $\sum_i L_i^\dagger L_i$ must converge to a bounded operator.
In this paper, we restrict our attention to the case when $\hilb_S$ has finite dimension and the number of the Lindblad operators $L_i$ is finite.
The evolution of $\ob{\rho}(t)$ is then described by an ordinary differential equation called the GKSL equation, after those who originally described it,
$$\dfrac{d}{dt}\ob{\rho}(t)=\lind\left[\ob{\rho}(t)\right],$$
also known simply as the master equation.  
\par
Solving the master equation directly amounts to solving a system of linear ordinary differential equations (ODE).  The number of these equations is of order $N^2$, where $N$ is the dimension of the system's Hilbert space.  For large $N$ this presents computational difficulties and a Monte Carlo approach called {\it unraveling} is preferred.   An unraveling is  a stochastic process $\ket{\psi_t(M)}$ valued in $\hilb_S$, with the property that the expected value of the one-dimensional projection operator
$$\ob{\rho}(t)=\Expc \state{\psi_t(M)}$$
satsifies the GKSL equation.  The known unravelings are solutions of stochastic differential equations (SDE), driven by a (multidimensional) noise process $M$---usually a Wiener process $W$
or a Poisson process $N$---more details are presented below.  Using an unraveling replaces integration of $N^2$ ODE by solving an $N$-dimensional SDE enough times to obtain a good 
approximation to the expected value representing the solution of the master equation.  
Two unravelings of the GKSL equation are well known and widely used.  One is the solution of the Gisin-Percival SDE \cite{GP92}, driven by a complex Wiener process $W(t)$ with components 
$W_i(t)$, whose number is equal to the number of the Lindblad operators, reads
\begin{align}
    d\ket{\psi_t}=-iH\ket{\psi_t}dt+&\sum_i\left( \expc{L_i^\dagger}_tL_i-\frac{1}{2}L^\dagger_i L_i-\frac{1}{2}|\expc{L_i}_t|^2
    \right)\ket{\psi_t}dt\nonumber\\
    &+\frac{1}{\sqrt{2}}\sum_i \left(L_i-\expc{L_i}_t\right)\ket{\psi_t}dW_i(t)
    \label{GPeqn}
\end{align}
 Here we've abbreviated $\expc{\psi_t|\cdot|\psi_t}$ as $\expc{\cdot}_t$. Another unravelling, known as the Piecewise Determinstic process (PDP) (see e.g. \cite{BP02}) is a solution of the SDE
\begin{align}
    d\ket{\psi_t}=-\bigg(iH&+\frac{1}{2}\sum_i L_i^\dagger L_i-\expc{L_i^\dagger L_i}_{t-}\bigg)\ket{\psi_{t-}}dt\nonumber\\
    &+\sum_i \left(\dfrac{L_i}{\langle L_i^\dagger L_i\rangle_{t-}^{1/2}}-I\right)\ket{\psi_{t-}}dN_i(t),
    \label{PDP}
\end{align}
where each $N_i(t)$ is an inhomogeneous Poisson process, satisfying
$$\Expc dN_i(t)=\Expc \expc{L^\dagger_i L_i}_tdt.$$
  Note that, as the Poisson process $N$ is discontinuous, so are the realizations of the PDP process which are solutions of the 
SDE driven by $N$. Moreover, these realizations have jump discontinuities at the same values of $t$, where $N$ jumps.  It is thus essential to make clear whether the realizations of $N$ and of the 
PDP process are continuous from the left, or from the right (different conventions are adopted by different authors) and to take left or right limits of the integrands where appropriate.
We assume that the Poisson process $N(t)$ is continuous on the right, and an integral with respect to $N(t)$ is defined as going up to time $t$, {\it including} $t$, that is, 
$\int_0^t f(t)\,dN(t)$ is interpreted as the integral over the interval $(0,t]$. 
On the other hand, since the realizations of the Wiener process and of the solutions of the Gisin-Percival equation, are continuous functions of time, their values at $t$ and their limits as 
$s \to  t-$ coincide, so no additional conventions are necessary.   
\par
Up to now we have only presented unravelings as a tool to integrate the GKSL equation, but there are a class of master
equations which are used in stochastic collapse models where the status of single trajectories of an unraveling---real or fictitious---becomes important. If we assume that 
quantum systems are described by single trajectories, $\ket{\psi_t(\omega)}$ where $\omega$ is a realization of a stochastic process, rather than a linear 
Schr\"odinger evolution, or evolution under a master equation, then these trajectories will have features particular to their status as solutions to stochastic differential equations. 
Expected values of nonlinear functionals of $\ket{\psi_t}$ depend on the choice of an unraveling in general. If such quantities are observable, their experimental values would help choose the unraveling which is closer to physical reality. 

In the sequel, we investigate two types of nonlinear functionals. The first is the square of the probability for obtaining a certain eigenvalue of an observable, which is enough to determine the ensemble variance in the outcome of any measurement. The second is a type of quantum entropy and is physically important in the context of thermodynamics and an illustrative quantification of localization effects in stochastic collapse theories.
\section{Essential Elementary Stochastic Calculus}
A stochastic process (see e.g. \cite{Prot04}) is  a set of random variables $\left\{M_t\right\}$, indexed over time $t\in[0,T]$, which is adapted to a certain filtration $\left\{\mathcal{F}_t\right\}$ describing the set of events of the process which can be determined up to time $t$.  
The realizations of a stochastic process may be continuous as functions of $t$, as in the case of the Wiener process, or discontinuous, as in the case of the Poisson process. Stochastic calculus, which describes integration and differentiation of these processes, differs from classical calculus in important ways. Typically, a process does not have finite variation, so the usual Lebesgue-Stieltjes theory of integration does not apply. Integration against a stochastic processes is first defined on simple-processes which are piecewise constant in time, and then extended to more general processes, called semimartingales.  A semimartingale is the sum of a finite variation process and a martingale such as the Wiener process. A stochastic differential of a process is then defined, inverting the relation between a differential and its integral. This results in a calculus with a modified chain rule, described by It\^o's formula.
\begin{theorem}[It\^o's Formula]
	Let $f:\mathbb{R}^n\rightarrow\mathbb{R}$ be twice continuously differentiable and let $X_t=(X_t^1,\ldots,X_t^n)$ be an n-tuple of semimartingales. Then $f(X_t)$ is a semimartingale with
	\begin{align*}
		f(X_t)-f(X_0)=&\sum_{i=1}^n\int_{0+}\dfrac{\partial f}{\partial x_i}(X_{s-})dX^i_s+\frac{1}{2}\sum_{1\le i,j\le n}^n\int_{0+}\dfrac{\partial^2 f}{\partial x_i\partial x_j}(X_{s-})d\quadvar{X^i,X^j}
		^c_s\\+&\sum_{0\le s\le t}\left(f(X_s)-f(X_{s-})-\sum_{i=1}^n\dfrac{\partial f}{\partial x_i}(X_{s-})\Delta X^i_s\right)
	\end{align*}
	where $\quadvar{X,X}^c_t$ is the continuous part of the quadratic variation process and $\Delta X_s=X_s-X_{s-}$ is the jump part of the semimartingale.
	\label{Itoform}
\end{theorem}
Here the quadratic variation is defined as 
$$\quadvar{X,Y}_t=X_tY_t - \dint{0}{t}X_{s-}dY_s-\dint{0}{t}Y_{s-}dX_s.$$
It\^o's formula tells us how to calculate a function of an arbitrary number of semimartingales, but a special case is often enough.
\begin{corollary}[It\^o's rule]
	The product of two semimartingales $X_t$ and $Y_t$ can be expressed in differential form
	$$d(X_tY_t)=(dX_t)Y_t+X_tdY_t+(dX_t)(dY_t),$$
	where the products of the differentials is calculated using the It\^o rules:
	\begin{align*}
		dW_i(t) dW_j(t)&=dW_i^*(t)dW_j^*(t)=0,\\
		dW_i(t)dt &= 0\\
		dW_i^*(t)dW_j(t)&=2\delta_{ij}dt,
		\label{witorule}
	\end{align*}
	for the Wiener case and
	\begin{align*}
		dN_i(t)dt &= 0,\\
		dN_i(t) dN_j(t)&=\delta_{ij}dN_i(t)
	\end{align*}
	for the Poisson case.
	\label{Itorule}
\end{corollary}
The so called `It\^o correction' to the product rule comes from the quadratic variation, where we write the heuristic $dXdY$ for $d\quadvar{X,Y}$.
This rule is sufficient to show that the Gisin-Percival equation and the PDP are unravellings by taking the expected value of
$$d\state{\psi_t}=\left(d\ket{\psi_t}\right)\bra{\psi_t}+\ket{\psi_t}\left(d\bra{\psi_t}\right)+\left(d\ket{\psi_t}\right)\left(d\bra{\psi_t}\right).$$
The It\^o correction is also the source of the discrepancy between evolution described by a master equation and its (different) unravelings. The simplest case where this difference can be seen is the variance.
\section{Variance}
Let an observable $A$ have spectral decomposition $A_i=\sum_i \lambda_iP_i$, and let $p_i=\expc{P_i}$ be the probability of obtaining the eigenvalue $\lambda_i$, so that the mean of $A$ equals $\expc{A\,}=\sum_i\lambda_ip_i$.
In an unravelling, this probability becomes stochastic, but it does so in such a way that the ensemble average of $p_i$ coincides with the value predicted in the standard quantum mechanics of open systems (of which quantum mechanics is a special case). The square of the expected value of an observable is then
$$\expc{A\,}^2=\sum_{ik} \lambda_i\lambda_k p_ip_k.$$
Using these rules, we can show that the evolution of probabilities under the Gisin-Percival equation is given by
\begin{equation}
	dp_i(t)=\expc{\lind^\dagger\left[P_i\right]}_tdt+\dfrac{1}{\sqrt{2}}\left[\sum_j\expc{P_i\left(L_j-\expc{L_j}\right)}_tdW_j(t) +h.c.\right],
	\label{pw}
\end{equation}
where $h.c.$ stands for Hermitian conjugate, and for the PDP the corresponding equation is
\begin{align}
	dp_i(t)&=\expc{\lind^\dagger\left[P_i\right]}_tdt+\sum_j\left( p_i\expc{L_j^\dagger L_j}_t-\expc{L_j^\dagger P_i L_j}_t\right)dt\nonumber\\
	&+\sum\left(\dfrac{\expc{L_j^\dagger P_i L_j}_{t-}}{\expc{L_j^\dagger L_j}_{t-}}-p_i(t-)\right)dN_j(t).
	\label{pn}
\end{align}
Here $\lind^\dagger$ is the adjoint of the Lindbladian, which is defined as the (unique) superoperator satisfying
$$\Tr\left[A^\dagger\lind[B]\right]=\Tr\left[\lind^\dagger[A^\dagger] B\right],$$
for any two operators $A$ and $B$. In terms of the Lindblad operators, it takes the form
$$\lind^\dagger[A]=i[H,A]-\dfrac{1}{2}\sum_iL_i^\dagger L_iA+AL_i^\dagger L_i-2L_i^\dagger A L_i$$
The Lindbladian evolution describes an analog of Schr\"odinger evolution for the density operator.  Hence the adjoint of the Lindbladian generates an analog of the Heisenberg evolution, describing how observables (in this case $P_i$) evolve.
Note that both equations \ref{pw} and \ref{pn} consist of the usual (adjoint) Lindbladian evolution plus a martingale term which has expected value zero. This is manifest in the case of equation
\ref{pw}; to exhibit this structure in the case of equation \ref{pn} we make use of the compensated Poisson process, $\wt{N}$, which is a martingale defined by subtracting from the Poisson process its expected
value. In this case the compensated Poisson process has the differential
$$d\wt{N}(t)=dN(t)-\expc{L_j^\dagger L_j}_tdt$$
and we can rewrite equation \ref{pn} as
\begin{align}
	dp_i(t)&=\expc{\lind^\dagger\left[P_i\right]}_tdt+\sum_j\left(\dfrac{\expc{L_j^\dagger P_i L_j}_{t-}}{\expc{L_j^\dagger L_j}_{t-}}-p_i(t-)\right)d\wt{N}_j(t).
	\label{pn}
\end{align}
The important consequence is that, as the martingales have expected value zero, on average the probabilities 
are what we would expect from the Lindbladian evolution. However, these martingales lead to nontrivial contributions to the second moment and the entropy functional because of Theorem \ref{Itoform}.  After taking the expected value, these corrections to the Lindbladian evolution do not disappear.
\par
To calculate the expected value of the second moment we apply the It\^o rule to the product $p_ip_k$. We find in the Wiener case that
\begin{align}
	d(p_ip_k) &= \left[\expc{\lind[P_i]}_tdt+\dfrac{1}{\sqrt{2}}\left(\sum_j\expc{P_i\left(L_j-\expc{L_j}\right)}_tdW_j+\mbox{h.c.}\right)\right]p_k\\\nonumber
	&+ p_i\left[\expc{\lind[P_k]}_tdt+\dfrac{1}{\sqrt{2}}\left(\sum_j\expc{P_k\left(L_j-\expc{L_j}\right)}_tdW_j+\mbox{h.c.}\right)\right]\\\nonumber
	& +\left[\sum_j\expc{P_i\left(L_j-\expc{L_j}_t\right)}_t\expc{\left(L^\dagger-\expc{L_j^\dagger}_t\right)P_k}_t+\right.\\\nonumber
	&+\left.\sum_j\expc{\left(L_j^\dagger-\expc{L_j^\dagger}_t\right)P_i}_t\expc{P_k\left(L_j-\expc{L_j}_t\right)}_t\right]dt.
\end{align}
If we take the ensemble average, the martingale terms disappear and we are left with
\begin{align}
	d\Expc p_ip_k &= \left[\Expc p_i\expc{\lind[P_k]}_t+p_k\expc{\lind[P_i]}_t\right]dt\\\nonumber
	& +\left[\sum_j\Expc\expc{P_i\left(L_j-\expc{L_j}\right)}_t\expc{\left(L^\dagger-\expc{L_j^\dagger}_t\right)P_k}_t+\right.\\\nonumber
	&+\left.\sum_j\Expc\expc{\left(L_j^\dagger-\expc{L_j^\dagger}_t\right)P_i}_t\expc{P_k\left(L_j-\expc{L_j}_t\right)}_t\right]dt.
\end{align}
exhibiting the corrections arising from the It\^o rule. The analogous calculation for the
Poisson case yields
\begin{align}
	d(p_ip_k) &= \left(\expc{\lind^\dagger\left[P_i\right]}_tdt+\sum_j\left(\dfrac{\expc{L_j^\dagger P_i L_j}_{t-}}{\expc{L_j^\dagger L_j}_{t-}}-p_i(t-)\right)d\wt{N}_j(t)\right)p_k+\\\nonumber
	& + p_i\left(\expc{\lind^\dagger\left[P_k\right]}_tdt+\sum_j\left(\dfrac{\expc{L_j^\dagger P_k L_j}_{t-}}{\expc{L_j^\dagger L_j}_{t-}}-p_k(t-)\right)d\wt{N}_j(t)\right)\\\nonumber
	& + \sum_j\left(\dfrac{\expc{L_j^\dagger P_i L_j}_{t-}}{\expc{L_j^\dagger L_j}_{t-}}-p_i(t-)\right)\left(
	\dfrac{\expc{L_j^\dagger P_k L_j}_{t-}}{\expc{L_j^\dagger L_j}_{t-}}-p_k(t-)\right)dN_i,
\end{align}
and taking the ensemble average we obtain
\begin{align}
	d\Expc p_ip_k & = \Expc\left[p_k\expc{\lind[P_i]}_t+p_i\expc{\lind[P_k]}_t\right]dt\\
	& + \Expc\left[\sum_j\left(\dfrac{\expc{L_j^\dagger P_i L_j}_t}{\expc{L_j^\dagger L_j}_t}-p_i(t)\right)\left(\dfrac{\expc{L_j^\dagger P_k L_j}_t}{\expc{L_j^\dagger L_j}_t}-p_k(t)\right)\expc{L_j^\dagger L_j}_t\right]dt.
\end{align}
Again, there is an It\^o correction. Crucially \textit{it is different from the Wiener case}. The second moment would manifest as an intrinsic dispersion (as opposed to the dispersion due to laboratory effects) of measurement values. The two unravellings thus give rise to two different variances for the value of an observable. See \cite{Petal23} for a detailed analysis of the variances for a two-level atom, including experimental relevance.
\par
It has been claimed that, if there existed a way of calculating this variance from experimental data, it would introduce the possibility for signaling faster than the speed of light \cite{Gisin89,BH15}. According to this point of view, it may not be physically possible to measure these moments.  On the other hand, in the stochastic collapse model where the trajectories are real and unravel a master equation, these modified moments are an unavoidable consequence.  Since superluminal signaling in this scenario requires the ability to influence which unravelling is being used, the problem is resolved by assuming that, in a fundamental sense, only one unravelling is possible and can be experimentally distinguished from its alternative. In this scenario, the Poisson unravelling could be considered as the more basic unravelling with the Wiener unravelling as its limiting case \cite{BB91}. The appearance of two unravellings is due entirely to the choice of measurement apparatus and no superluminal signaling is possible.

\section{Entropy in Unravelings}

In accordance with the measurement postulate of quantum mechanics (see e.g. \cite{NielChua00}), a measurement is a collection of operators, $\{M_i\}$ satisfying a 
completeness relation

$$\sum_i M_i^\dagger M_i=I.$$
The operator $M_i$ corresponds to a measurement result of type $i$, occurring with probability 
$$p_i=\expc{M_i^\dagger M_i}_t=\expc{P_i}_t,$$
where $P_i=M_i^\dagger M_i.$
Given an unraveling, we may ask: how does the entropy of this probability distribution change in time? The entropy depends on the
probabilities $p_i$ in the usual way
$$S=-\sum_i p_i\log p_i,$$
In the context of an unraveling, the $p_i$ are stochastic processes, as described above. Because entropy is a nonlinear function of the $p_i$, It\^o corrections will appear. 
\par
To see this we will first calculate the It\^o differential of $\log p_i$. Using Theorem \ref{Itoform}, we obtain
\begin{align*}
d\log p_i(t)&=\dfrac{1}{p_i(t)}\left[\expc{\lind^\dagger\left[P_i\right]}_tdt+\frac{1}{\sqrt{2}}\left(\sum_j\expc{P_i\left(L_j-\expc{L_j}\right)}_tdW_j(t)+h.c\right)\right]\\
&-\frac{1}{p_i^2(t)}\sum_j\bigg|\expc{P_i\left(L_j-\expc{L_j}\right)}_t\bigg|^2dt
\end{align*}
in the Wiener case, and
\begin{align*}
d\log p_i(t)&=\dfrac{1}{p_i(t)}\left[\expc{\lind^\dagger\left[P_i\right]}_tdt-\sum_j\left(\expc{L_j^\dagger P_iL_j}_t-p_i(t)\expc{L_j^\dagger L_j}_t\right)dt\right]\\
&+\sum_j\left[\log\left(\dfrac{\expc{L_j^\dagger P_iL_j}_{t-}}{\expc{L_j^\dagger L_j}_{t-}}\right)-\log p_i(t-)\right]dN_j(t)
\end{align*}
in the Poisson case.
We may then calculate the stochastic differential of the entropy functional by applying Corollary \ref{Itorule} to $p_i$ and $\log p_i$. The results are
\begin{align}
dS(t) &= -\sum_i\expc{\lind^\dagger[P_i]}_t\log p_i(t)dt-\dfrac{1}{\sqrt{2}}\sum_{ij}\left[\expc{P_i\left(L_j-\expc{L_j}\right)}_t\log p_i(t) dW_j(t) +h.c.\right]\nonumber\\
&-\sum_i\dfrac{1}{p_i(t)}\sum_j\bigg|\expc{P_i\left(L_j-\expc{L_j}\right)}_t\bigg|^2dt,
\label{sw}
\end{align}
for the Wiener case and
\begin{align}
dS(t) &=-\sum_i\expc{\lind^\dagger[P_i]}_t\log p_i(t)dt-\sum_{ij}\Biggl[\dfrac{\expc{L_j^\dagger P_iL_j}_{t-}}{\expc{L_j^\dagger L_j}_{t-}}\log\left(\dfrac{\expc{L_j^\dagger P_iL_j}_{t-}}{\expc{L_j^\dagger L_j}_{t-}}\right)\nonumber \\
&-p_i(t-)\log p_i(t-)\Biggr]d\wt{N}_j(t)
-\sum_{ij}\expc{L_j^\dagger P_iL_j}_t\log\left(\dfrac{\expc{L_j^\dagger P_iL_j}_t}{\expc{L_j^\dagger L_j}_tp_i(t)}\right)dt
\label{sn}
\end{align}
for the Poisson case.
Here, simplifications have been made by using the completeness relation, $\sum_iP_i=I$, and the fact that $\lind^\dagger[I]=0$, which is a consequence of the preservation of trace by the Lindbladian evolution.
We see that there are non-martingale corrections to the entropy, and so making use of the unravellings themselves---viewing them as real trajectories instead of just using the expectation to
unravel the Lindblad equation---has real consequences for the entropy.  The same is true of other nonlinear functionals of quantum trajectories.
\par
It is important to point out that this is by no means the only entropy that can be defined. There is of course the von Neumann entropy, $S^{vN}(t)$, which is defined as
\begin{equation}
	S^{vN}(t)=-\sum\Tr[\overline{\rho}(t)\log\overline{\rho}(t)].
\end{equation} 
A closely related quantity, is the average entropy associated with the measurement of an observable $A$, which evolves according to the adjoint Lindbladian evolution:
$$S^A(t)=-\sum_i\left(\Expc p_i(t)\right)\log\left(\Expc p_i(t)\right),$$
We compare $S^A$ to the mean of $S$ in the next section.
\section{Non-martingale Corrections to the Entropy}
The average entropy for an arbitrary unravelling is smaller than the measurement entropy, $S^A$, obtained from the master equation. To see this we simply apply Jensen's inequality to the probabilties with convex function $\phi(\xi) = \xi\log \xi$ to
get that $\Expc\phi(p_i)\ge \phi(\Expc p_i)$ and so $\Expc S_t\le S^{A}$.
$S^{A}$ obeys a deterministic evolution under the master equation
$$dS^{A} = \sum_i\Expc\left[\expc{\lind^\dagger[P_i]}_t\right]\log\left( \Expc p_i(t)\right)dt$$
The expected value of the stochastic process $S$ differs from the deterministic evolution of $S^{A}$ in three crucial ways. The first is that it is possible to have correlations between $\expc{\lind^\dagger[P_i]}$ and $\log p_i$ so that in general the expected value of their product is not equal to the product of their expected values. The second difference is that there are martingale corrections which average to zero but contribute to the dynamics of single realizations. The third are the It\^o  corrections which can be shown to be manifestly negative. In the Wiener case, this follows from inspection. We show that the correction is negative for the Poisson case below.
\begin{proposition}
Let $f_t=\sum_{ij}\expc{L_j^\dagger P_iL_j}_{t}\log\left(\dfrac{\expc{L_j^\dagger P_i L_j}_t}{\expc{L_j^\dagger L_j}_t\expc{P_i}}_t\right)$ and suppose 
$\sum_j L_j^\dagger L_j=B$. Then $f_t\ge 0$.
\end{proposition}
\begin{proof}
Fix a time $t$ and realization for $\ket{\psi_t}$. Define a discrete random variable $X$ taking values 
$x_{ij}=\dfrac{\expc{L_j^\dagger P_iL_j}_t}{\expc{L_j^\dagger L_j}_t\expc{P_i}_t}$ with probabilities 
$p_{ij}=\dfrac{1}{\expc{\,B\,}}\expc{L_j^\dagger L_j}_t\expc{P_i}_t$. Note that $B$ has positive expectation by definition. Let $EX=\sum_{ij}x_{ij}p_{ij}$ and $\phi(\xi)=\xi\log \xi$. As $\phi(\xi)$ is a convex function, Jensen's inequality implies that
$$E\phi(X)\ge \phi(EX)$$
We find that $EX=\dfrac{1}{\expc{\,B\,}}\sum_{ij}\expc{L_j^\dagger P_i L_j}_t=1$, and hence $\phi(EX)=0$. Thus
$$E\phi(X)=f_t\ge 0$$
\end{proof}
The fact that entropy of the unravelings is lower than the entropy $S^A$ calculated from the Lindbladian evolution allows for the possibility that the two evolutions can be distinguished experimentally.
One can imagine heating a quantum system close to absolute zero and quantifying the amount of heat absorbed. The ability to absorb heat is reflected in the entropy, and is lower in localized systems.
\section{Conclusion}
It seems natural that the process of localization should decrease the accessible states of the system but localization can also increase the energy. This can result in a situation where the entropy is decreasing but the energy is increasing, a kind of `negative temperature' scenario for the trajectory-wise entropy. The von Neumann entropy may increase in accordance with the dynamics prescribed by the master equation, and in the unraveling picture this corresponds to a situation where one has an increasingly broadly distributed classical ensemble of localized wavefunctions.
\par
On the other hand one can argue that predictions from the trajectory-wise picture should not be experimentally detectable since we live on a single trajectory according to stochastic collapse models. This means that we cannot measure the trajectory-wide variance and thus the variance cannot be used to transmit signals faster than the speed of light. In practice, it is often assumed that an identically prepared system is identical and so numerous realizations can be generated by simply preparing the same system many times, but this might be too simplistic an idea, and since the alternative leads to possible faster-than-light signaling it should be taken seriously. Faster-than-light signaling would require some choice to made which can be communicated through a channel, thus the complication is avoided if we take one unravelling, the Poisson one, to be the underlying unraveling with the Wiener unravelling appearing through choice of measurement apparatus, a choice which cannot be used to signal since it will only have local influence. 
\par
Either way the use of the trajectory model in stochastic collapse theories forces us to reckon with the It\^o corrections calculated above.

\bibliography{entropy_240205}
\end{document}